\documentclass[conference,10pt]{IEEEtran}
\makeatletter
\def\ps@headings{%
\def\@oddhead{\mbox{}\scriptsize\rightmark \hfil \thepage}%
\def\@evenhead{\scriptsize\thepage \hfil\leftmark\mbox{}}%
\def\@oddfoot{}%
\def\@evenfoot{}}
\makeatother
\usepackage{verbatim,amssymb,amsmath,graphicx}
\usepackage{subfigure}
\IEEEoverridecommandlockouts

\newcommand{\FF}{\mathbb{F}}
\newcommand{\MM}{\mathbb{M}}

\newcommand{\ov}{\mathbf{o}}
\newcommand{\pv}{\mathbf{p}}
\newcommand{\sv}{\mathbf{s}}
\newcommand{\xv}{\mathbf{x}}
\newcommand{\yv}{\mathbf{y}}

\newcommand{\hsrc}{{\rm HSRC}}
\newcommand{\nmax}{n_{max}}
\newcommand{\probup}{p_{frag}}

\newcommand{\objup}{p_{obj}}
\newcommand{\divsrc}{\delta} % Number of distinct and mutually exclusive pairs of fragments which can be used to reconstruct a specific fragment. Depends on the diversity of the symmetric subset code.

\newtheorem{ex}{Example}
\newtheorem{lem}{Lemma}
\newtheorem{cor}{Corollary}
\newtheorem{defn}{Definition}

\begin{document}
\title{Self-repairing Homomorphic Codes \\for Distributed Storage Systems\thanks{F. Oggier's research is supported in part by the Singapore National Research Foundation under Research Grant NRF-RF2009-07 and NRF-CRP2-2007-03,
and in part by the Nanyang Technological University under research grant M58110049 and M58110070. A. Datta's research for this work has been supported by AcRF Tier-1 grant number RG 29/09.}}
\author{
\IEEEauthorblockN{Fr\'ed\'erique Oggier}
\IEEEauthorblockA{Division of Mathematical Sciences \\
School of Physical and Mathematical Sciences\\
Nanyang Technological University\\ Singapore\\
Email: frederique@ntu.edu.sg}
\and
\IEEEauthorblockN{Anwitaman Datta}
\IEEEauthorblockA{Division of Computer Science\\
School of Computer Engineering\\
Nanyang Technological University\\ Singapore\\
Email: anwitaman@ntu.edu.sg}
}
\maketitle

%*************************************************************************%
%
% ABSTRACT
%
%************************************************************************%

\begin{abstract}
Erasure codes provide a storage efficient alternative to replication based redundancy in (networked) storage systems. They however entail high communication overhead for maintenance, when some of the encoded fragments are lost and need to be replenished. Such overheads arise from the fundamental need to recreate (or keep separately) first a copy of the whole object before any individual encoded fragment can be generated and replenished. There has been recently intense interest to explore alternatives, most prominent ones being regenerating codes (RGC) and hierarchical codes (HC). We propose as an alternative a new family of codes to improve the maintenance process, which we call \emph{self-repairing codes} (SRC), with the following salient features: (a) encoded fragments can be repaired directly from other subsets of encoded fragments without having to reconstruct first the original data, ensuring that (b) a fragment is repaired from a fixed number of encoded fragments, the number depending only on how many encoded blocks are missing and independent of which specific blocks are missing. These properties allow for not only low communication overhead to recreate a missing fragment, but also independent reconstruction of different missing fragments in parallel, possibly in different parts of the network. The fundamental difference between SRCs and HCs is that different encoded fragments in HCs do not have symmetric roles (equal importance). Consequently the number of fragments required to replenish a specific fragment in HCs depends on which specific fragments are missing, and not solely on how many. Likewise, object reconstruction may need different number of fragments depending on which fragments are missing. RGCs apply network coding over $(n,k)$ erasure codes, and provide network information flow based limits on the minimal maintenance overheads. RGCs need to communicate with at least $k$ other nodes to recreate any fragment, and the minimal overhead is achieved if only one fragment is missing, and information is downloaded from all the other $n-1$ nodes. We analyze the \emph{static resilience} of SRCs with respect to traditional erasure codes, and observe that SRCs incur marginally larger storage overhead in order to achieve the aforementioned properties. The salient SRC properties naturally translate to \emph{low communication overheads} for reconstruction of lost fragments, and allow reconstruction with lower latency by facilitating \emph{repairs in parallel}. These desirable properties make self-repairing codes a good and practical candidate for networked distributed storage systems.
\end{abstract}
\section{Introduction}

Networked storage systems have gained prominence in recent years. These include various genres, including decentralized peer-to-peer storage systems, as well as dedicated infrastructure based data-centers and storage area networks. Because of storage node failures, or user attrition in a peer-to-peer system, redundancy is essential in networked storage systems. This redundancy can be achieved using either replication, or (erasure) coding techniques, or a mix of the two. Erasure codes require an object to be split into $k$ parts, and mapped into $n$ encoded fragments, such that any $k$ encoded fragments are adequate to reconstruct the original object. Such coding techniques play a prominent role in providing storage efficient redundancy, and are particularly effective for storing large data objects and for archival and data back-up applications (for example, CleverSafe \cite{cleversafe}, Wuala \cite{wuala}).

Redundancy is lost over time because of various reasons such as node failures or attrition, and mechanisms to maintain redundancy are essential. It was observed in \cite{Liskov} that while erasure codes are efficient in terms of storage overhead, maintenance of lost redundancy entail relatively huge overheads. A naive approach to replace a single missing fragment will require that $k$ encoded fragments are first fetched in order to create the original object, from which the missing fragment is recreated and replenished. This essentially means, for every lost fragment, $k$-fold more network traffic is incurred when applying such a naive strategy.

Several engineering solutions can partly mitigate the high maintenance overheads. One approach is to use a `hybrid' strategy, where a full replica of the object is additionally maintained \cite{Liskov}. This ensures that the amount of network traffic equals the amount of lost data.\footnote{In this paper, we use the terms `fragment' and `block' interchangeably. Depending on the context, the term `data' is used to mean either fragment(s) or object(s).} A spate of recent works \cite{netcod,hierarchical} argue that the hybrid strategy adds storage inefficiency and system complexity. Another possibility is to apply lazy maintenance \cite{TotalRecall,dattaP2P}, whereby maintenance is delayed in order to amortize the maintenance of several missing fragments. Lazy strategies additionally avoid maintenance due to temporary failures. Procrastinating repairs however may lead to a situation where the system becomes vulnerable, and thus may require a much larger amount of redundancy to start with. Furthermore, the maintenance operations may lead to spikes in network resource usage \cite{dattaSSS}.

It is worth highlighting at this juncture that erasure codes had originally been designed in order to make communication robust, such that loss of some packets over a communication channel may be tolerated. Network storage has thus benefitted from the research done in coding over communication channels by using erasure codes as black boxes that provide efficient distribution and reconstruction of the stored objects. Networked storage however involves different challenges but also opportunities not addressed by classical erasure codes. Recently, there has thus been a renewed interest \cite{netcod,UCBsubm,hierarchical,biersackRGC,vijaykumarallerton} in designing codes that are optimized to deal with the vagaries of networked storage, particularly focusing on the maintenance issue. In a volatile network where nodes may fail, or come online and go offline frequently, new nodes must be provided with fragments of the stored data to compensate for the departure of nodes from the system, and replenish the level of redundancy (in order to tolerate further faults in future). In this paper, we propose a new family of codes called
\emph{self-repairing codes} (SRC), which are tailored to fit well typical networked storage environments.

%************************************************************************%
\subsection{Related work}

In \cite{netcod,UCBsubm}, Dimakis et al. propose regenerating codes (RGC) by
exposing the need of being able to reconstruct an erased encoded block from a
smaller amount of data than would be needed to first reconstruct the whole
object. They however do not address the problem of building new codes that
would solve the issue, but instead use classical erasure codes as a black box
over a network which implements random linear network coding and propose leveraging the properties of network
coding to improve the maintenance of the stored data. Network information flow based analysis shows the possibility to replace missing fragment using network traffic equalling the volume of lost data. Unfortunately, it is possible to achieve this optimal limit only by communicating with all the $n-1$ remaining blocks. Consequently, to the best of our knowledge, regenerating codes literature generally does not discuss how it compares with engineering solutions like lazy repair, which amortizes the repair cost by initiating repairs only when several fragments are lost. Furthermore, for RGCs to work, even sub-optimally, it is essential to communicate with at least $k$ other nodes to reconstruct any missing fragment. Thus, while the volume of data-transfer for maintenance is lowered, RGCs are expected to have higher protocol overheads, implementation and computational complexity. For instance, it is noted in \cite{biersackRGC} that a randomized linear coding based realization of RGCs takes an order of magnitude more computation time than standard erasure codes for both encoding and decoding.
The work of \cite{vijaykumarallerton} improves on the original RGC papers in that instead of arguing the existence of regenerating codes via deterministic network coding algorithms, they provide explicit network code constructions.

In \cite{hierarchical}, the authors make the simple observation that encoding two bits into three by XORing the two information bits has the property that any two encoded bits can be used to recover the third one. They then propose an iterative construction where, starting from small erasure codes, a bigger code, called hierarchical code (HC), is built by XORing subblocks made by erasure codes or combinations of them. Thus a subset of encoded blocks is typically enough to regenerate a missing one. However, the size of this subset can vary, from the minimal to the maximal number of encoded subblocks, determined by not only the number of lost blocks, but also the specific lost blocks. So given some lost encoded blocks, this strategy may need an arbitrary number of other encoded blocks to repair.

%*************************************************************************%
\subsection{Self Repairing Codes}

While motivated by the same problem as RGCs and HCs, that of efficient
maintenance of lost redundancy in coding based distributed storage systems,
the approach of self-repairing codes (SRC) tries to do so at a somewhat
different point of the design space. We try to minimize the number of nodes
necessary to reduce the reconstruction of a missing block, which automatically
translates into lower bandwidth consumption, but also lower computational
complexity of maintenance, as well as the possibility for
faster and parallel replenishment of lost redundancy.

We define the \emph{concept of self-repairing codes} as
$(n,k)$ codes designed to suit networked storage systems, that encode $k$
fragments of an object into $n$ encoded fragments to be stored at $n$ nodes,
with the properties that:\\
(a) \emph{encoded fragments can be repaired directly from other subsets
of encoded fragments without having to reconstruct first the original data}. \\
More precisely, based on the analogy with the error correction capability of erasure codes, which is of any $n-k$ losses independently of which losses,\\ (b) \emph{a fragment can be repaired from a fixed number of encoded fragments, the number depending only on how many encoded blocks are missing and independent of which specific blocks are missing.} %in order to improve the maintenance of data storage.

To do so, SRCs naturally require more redundancy than erasure codes. We will
see more precisely later on that there is a tradeoff between the repair ability
and this extra redundancy. Consequently, SRCs can recreate the whole object
with $k$ fragments, though unlike for erasure codes, these are not arbitrary
$k$ fragments, though many such $k$ combinations can be found (see Section
\ref{sec:static} for more details).

Note that even for traditional erasure codes, the property (a) may
coincidentally be satisfied, but in absence of a systematic mechanism this
serendipity cannot be leveraged. In that respect, HCs \cite{hierarchical}
may be viewed as a way to do so, and are thus the closest example of
construction we have found in the literature, though they do not give any
guarantee on the number of blocks needed to repair given the number of losses,
i.e., property (b) is not satisfied, and has no deterministic guarantee for
achieving property (a) either. We may say that in spirit, SRC is closest to
hierarchical codes - at a very high level, SRC design features mitigate the
drawbacks of HCs.

In this work, we make the following \emph{contributions}:\\
(i) We propose a new family of codes, self-repairing codes (SRC), designed specifically as an alternative to erasure codes (EC) for providing redundancy in networked storage systems, which allow repair of individual encoded blocks using only few other encoded blocks. Like ECs, SRCs also allow recovery of the whole object using $k$ encoded fragments, but unlike in ECs, these are not any arbitrary $k$ fragments. However, numerous specific suitable combinations exist.\\
(ii) We provide a deterministic code construction called \emph{Homomorphic Self-Repairing Code} (HSRC), showcasing that SRC codes can indeed be realized.\\
(iii) HSRC self-repair operations are computationally efficient. It is done by XORing encoded blocks, each of them containing information about all fragments of the object, though the encoding itself is done through polynomial evaluation, not by XORing.\\
(iv) We show that for equivalent static resilience, marginally more storage is needed than traditional erasure codes to achieve self-repairing property.\\
(v) The need of few blocks to reconstruct a lost block naturally translates to low overall bandwidth consumption for repair operations. SRCs allow for both eager as well as lazy repair strategies for equivalent overall bandwidth consumption for a wide range of practical system parameter choices. They also outperform lazy repair with the use of traditional erasure codes for many practical parameter choices.\\
(vi) We show that by allowing parallel and independent repair of different encoded blocks, SRCs facilitate fast replenishment of lost redundancy, allowing a much quicker system recovery from a vulnerable state than is possible with traditional codes.

%*************************************************************************%
%
% LINEAR CODING THROUGH POLY
%
%*************************************************************************%
\section{Linear Coding through Polynomials}
\label{sec:lincod}

Since this work aims at designing specifically tailored codes for networked
storage systems, we first briefly recall the mechanisms behind erasure codes design.
In what follows, we denote by $\FF_q$ the finite field with $q$ elements, and
by $\FF_q^*$ the finite field without the zero element.
If $q=2^m$, an element $\xv\in\FF_q$ can be represented by an $m$-dimensional
vector $\xv=(x_1,\ldots,x_m)$ where $x_i\in\FF_2$, $i=1,\ldots,m$, coming from
fixing a basis , namely  $\xv=\sum_{i=1}^mx_iw^{i-1}$ where
$\{1,w,\ldots,w^{m-1}\}$ forms a $\FF_2$-basis of $\FF_q$, and $w$ is a root
of an irreducible monic polynomial of degree $m$ over $\FF_2$.
The finite field $\FF_2$ is nothing else than the two bits 0 and 1, with
addition and multiplication modulo 2.

%*************************************************************************%
\subsection{Erasure codes}

A linear $(n,k,d)$ erasure code over a $q$-ary alphabet is formally a linear
map $c:\FF_{q^k}\rightarrow\FF_{q^n},~\sv\mapsto c(\sv)$
which maps a $k$-dimensional vector $\sv$ to an $n$-dimensional vector
$c(\sv)$. The set $C$ of codewords $c(\sv)$, $\sv\in\FF_{q^k}$,
forms the code (or codebook). The third parameter $d$ refers to the minimum
distance of the code: $d=\min_{\xv\neq\yv \in C} d(\xv,\yv)$ where the Hamming distance $d(\xv,\yv)$ counts the number of positions at which the coefficients of $\xv$ and $\yv$ differ. The minimum distance describes how many erasures can
be tolerated, which is known to be at most $n-k$, achieved by maximum distance
separable (MDS) codes. MDS codes thus allow to recover any codeword out of $k$
coefficients.

Let $\ov$ be an object of size $M$ bits, that is $\ov\in\FF_{2^M}$, and let $k$ be a positive integer such that $k$ divides $M$.
We can write
\[
\ov=(\ov_1,\ldots,\ov_k),~\ov_i\in\FF_{2^{M/k}}
\]
which requires the use of a $(n,k)$ code over $\FF_{2^{M/k}}$, that
maps $\ov$ to an $Mn/k$-dimensional binary vector $\xv$, or equivalently, an $n$-dimensional vector
\[
\xv=(\xv_1,\ldots,\xv_n),~\xv_i\in\FF_{2^{M/k}}.
\]

%*************************************************************************%
\subsection{Reed-Solomon Codes}
\label{subsec:encodpoly}

Since the work of Reed and Solomon \cite{ReedSolomon}, it is known that linear
coding can be done via polynomial evaluation. In short, take an object
$\ov=(o_1,o_2,\ldots,o_k)$ of size $M$, with each $\ov_i$ in $\FF_{2^{M/k}}$,
and create the polynomial
\[
p(X)=\ov_1+\ov_2X+\ldots\ov_kX^{k-1}\in\FF_{2^{M/k}}[X].
\]
Now evaluate $p(X)$ in $n$ elements $\alpha_1,\ldots,\alpha_n \in\FF_{2^{M/k}}^*$,
to get the codeword
\[
(p(\alpha_1),\ldots,p(\alpha_n)),~n\leq 2^{M/k}-1.
\]
\begin{ex}\label{ex:exRS23} \rm
Suppose the object $\ov=(o_1,o_2,o_3,o_4)$ has 4 bits, and
we want to make $k=2$ fragments: $\ov_1 = (o_1,o_2)\in \FF_4$,
$\ov_2= (o_3,o_4)\in \FF_4$. We use a $(3,2)$ Reed-Solomon code over
$\FF_4$, to store the file in 3 nodes.
Recall that $\FF_4=\{ (a_0,a_1),~a_0,a_1\in \FF_2\}=
\{ a_0+a_1 w,~a_0,a_1\in \FF_2\}$
where $w^2 = w+1$. Thus we can alternatively represent each fragment as:
$\ov_1= o_1+o_2w\in \FF_4$, $\ov_2= o_3+o_4w\in \FF_4$.
The encoding is done by first mapping the two fragments into a polynomial
$p(X)\in\FF_4[X]$:
\[
p(X)=(o_1+o_2w)+(o_3+o_4w)X,
\]
and then evaluating $p(X)$ into the three non-zero elements of $\FF_4$, to
get a codeword of length 3:
\[
(p(1),p(w),p(w+1))
\]
where $p(1)=o_1+o_3+w(o_2+o_4)$, $p(w)=o_1+o_4+w(o_2+o_3+o_4)$,
$p(w^2)=o_1+o_3+o_4+w(o_2+o_3)$,
so that each node gets two bits to store:
$(o_1+o_3,o_2+o_4)$ at node 1, $(o_1+o_4,o_2+o_3+o_4)$ at node 2,
$(o_1+o_3+o_4,o_2+o_3)$ at node 3.
\end{ex}

%*************************************************************************%
%
% HOMOMORPHIC CODES
%
%*************************************************************************%
\section{Homomorphic Codes}

Encoding linearly data as explained in Section \ref{sec:lincod} can be done
with arbitrary polynomials. We now first describe a particular class of polynomials
that will play a key role in the construction of homomorphic codes, a class of
self-repairing codes presented in Subsection \ref{subsec:src}.

%*************************************************************************%
\subsection{Linearized polynomials}

Since we work over finite fields that contains $\FF_2$, recall that all
operations are done in characteristic 2, that is, modulo 2.
Let $a,b \in \FF_{2^m}$, for some $m\geq 1$. Then we have that
$(a+b)^2=a^2+b^2$ and consequently
\begin{equation}\label{eq:Frob}
(a+b)^{2^i}=\sum_{j=0}^{2^i} {2^i\choose j} a^j b^{2^i-j}= a^{2^i}+b^{2^i},~i\geq 1.
\end{equation}
Recall the definition of a linearized polynomial.
\begin{defn}
A {\em linearized polynomial} $p(X)$ over $\FF_q$, $q=2^m$, has the form
\[
p(X)=\sum_{i=0}^{k-1}p_iX^{q^i},~p_i\in\FF_q.
\]
\end{defn}
We now define a {\em weakly linearized polynomial} as
\begin{defn}
A {\em weakly linearized polynomial} $p(X)$ over $\FF_q$, $q=2^m$, has the form
\[
p(X)=\sum_{i=0}^{k-1}p_iX^{2^i},~p_i\in\FF_q.
\]
\end{defn}
We will see below why we chose this name.
We use the notation $k$ since later on it will indeed correspond to the
number of data symbols that can be encoded with the proposed scheme.
We start with a useful property of such polynomials.
\begin{lem}\label{lem:sr}
Let $a,b \in \FF_{2^m}$ and let $p(X)$ be a weakly linearized polynomial
given by $p(X)=\sum_{i=0}^{k-1}p_iX^{2^i}$.
We have
\[
p(a+b)=p(a)+p(b).
\]
\end{lem}
\begin{IEEEproof}
Note that if we evaluate $p(X)$ in an element $a+b \in \FF_{2^m}$, we get,
using (\ref{eq:Frob}), that
\[
p(a+b)= \sum_{i=0}^{k-1}p_i(a+b)^{2^i}= \sum_{i=0}^{k-1}p_i(a^{2^i}+b^{2^i})
=p(a)+p(b).
\]
\end{IEEEproof}

We can strengthen the above lemma by considering instead a polynomial $p(X)$ over $\FF_q$, $q=2^m$, of the form:
\[
p(X)=\sum_{i=0}^{k-1}p_iX^{s^i},~p_i\in\FF_q,
\]
where $s=2^l$, $1\leq l\leq m$ ($l=m$ makes $p(X)$ a linearized polynomial).
We now get:
\begin{lem}
Let $a,b \in \FF_{2^m}$ and let $p(X)$ be the polynomial given by
$p(X)=\sum_{i=0}^{k-1}p_iX^{s^i}$, $s=2^l$, $m \geq l \geq 1$.
We have
\[
p(ua+vb)=up(a)+vp(b),~u,v\in \FF_s.
\]
\end{lem}
\begin{IEEEproof}
If we evaluate $p(X)$ in $ua+vb$, we get
\[
p(ua+vb)
= \sum_{i=0}^{k-1}p_i(ua+vb)^{s^i} \\
= \sum_{i=0}^{k-1}p_i((ua)^{s^i}+(vb)^{s^i})
\]
again by (\ref{eq:Frob}), and
\[
p(ua+vb)
= \sum_{i=0}^{k-1}p_i(ua^{s^i}+vb^{s^i})\\
= u\sum_{i=0}^{k-1}p_ia^{2^i}+ v\sum_{i=0}^{k-1}p_ib^{2^i}
\]
using the property that $u^s=u$ for $u\in\FF_s$.
\end{IEEEproof}

%*************************************************************************%
\subsection{Self-repairing codes using weakly linearized polynomials}
\label{subsec:src}

We now mimic the way encoding works for Reed-Solomon codes (see Subsection
\ref{subsec:encodpoly}) for weakly linearized polynomials. Note that neither
the encoding nor the decoding process described below are actual efficient
algorithms. Implementations of these processes is a separate issue to be dealt with.

{\flushleft {\bf Encoding:}}
\begin{enumerate}
\item
Take an object $\ov$ of length $M$, with $k$ a positive integer that divides
$M$. Decompose $\ov$ into $k$ fragments of length $M/k$:
\[
\ov=(\ov_1,\ldots,\ov_k),~\ov_i\in\FF_{2^{M/k}}.
\]
\item
Take a linearized polynomial with coefficients in $\FF_{2^{M/k}}$
\[
p(X)=\sum_{i=0}^{k-1}p_iX^{2^i},
\]
and encode the $k$ fragments as coefficients, namely take $p_i=\ov_{i+1}$,
$i=0,\ldots,k-1$.
\item
Evaluate $p(X)$ in $n$ non-zero values $\alpha_1,\ldots,\alpha_n$ of
$\FF_{2^{M/k}}$ to get a $n$-dimensional codeword
\[
(p(\alpha_1),\ldots,p(\alpha_n)),
\]
and each $p(\alpha_i)$ is given to node $i$ for storage. In particular,
we need
\begin{equation}\label{eq:boundn}
n\leq 2^{M/k}-1.
\end{equation}
\end{enumerate}

{\flushleft {\bf Decoding:}}
\begin{enumerate}
\item
Given $k$ linearly independent fragments, the node that wants to reconstruct
the file computes $2^k-1$ linear combinations of the $k$ fragments, which
gives $2^k-1$ points in which $p$ is evaluated.
\item
Lagrange interpolation guarantees that it is enough to have $2^{k-1}+1$ points
(which we have since $2^k-1 \geq 2^{k-1}+1$ for $k\geq 2$) to reconstruct
uniquely the polynomial $p$ and thus the data file. This requires
\begin{equation}\label{eq:boundk}
2^{k-1}+1 \leq 2^{M/k}-1.
\end{equation}
\end{enumerate}

{\flushleft {\bf Self-repairing:}} A codeword constructed with the above
procedure is of the form $(p(\alpha_1),\ldots,p(\alpha_n))$, where each
coefficient is in $\FF_{2^{M/k}}$ and $k<n\leq 2^{M/k}-1$. We will denote by
$\nmax$ the maximum value that $n$ can take, namely $\nmax=2^{M/k}-1$.
We know that $\FF_{2^{M/k}}$ contains a $\FF_2$-basis $B=\{b_1,\ldots,b_{M/k}\}$
with $M/k$ linearly independent elements.  If  $n=2^{M/k}-1$,  the $\alpha_i$,
$i=1,\ldots,n$, can be expressed as $\FF_2$-linear combinations of the basis
elements, and we have from Lemma \ref{lem:sr} that
\[
\alpha_i=\sum_{j=1}^{M/k}\alpha_{ij}b_j,~\alpha_{ij}\in\FF_2
\Rightarrow
p(\alpha_i)=\sum_{j=1}^{M/k}\alpha_{ij}p(b_j).
\]
In words, that means that an encoded fragment can be obtained as a linear
combination of other encoded fragments. In terms of computational complexity,
this further implies that the cost of a block reconstruction is that of some
XORs (one in the most favorable case, when two terms are enough to reconstruct
a block, up to $k-1$ in the worst case).
On the other hand, if $\alpha_1,\ldots,\alpha_n$ are
contained in $B$, then the code has no self-repairing property.

For any choice of a positive integer $k$ that divides $M$, we work in the
finite field $\FF_{2^{M/k}}$. To do explicit computations in this finite field,
it is convenient to use the generator of the multiplicative group
$\FF_{2^{M/k}}^*=\FF_{2^{M/k}}\backslash\{0\}$, that we will denote by $w$.
A generator has the property that $w^{2^{M/k}-1}=1$, and there is no smaller
positive power of $w$ for which this is true.

\begin{ex}\label{ex:complete}\rm
Take a data file $\ov=(o_1,\ldots,o_{12})$ of $M=12$ bits, and choose
$k=3$ fragments. We have that $M/k=4$, which satisfies (\ref{eq:boundk}),
that is $2^2+1=5 \leq 2^4-1=15=\nmax$.

The file $\ov$  is cut into 3 fragments $\ov_1=(o_1,\ldots,o_4)$,
$\ov_2=(o_5,\ldots,o_8)$, $\ov_3=(o_9,\ldots,o_{12}) \in \FF_{2^4}$.
Let $w$ be a generator of the multiplicative group of $\FF_{2^4}^*$, such
that $w^4=w+1$. The polynomial used for the encoding is
\[
p(X)=\sum_{i=1}^4o_iw^iX+\sum_{i=1}^4o_{i+4}w^iX^2+\sum_{i=1}^4o_{i+8}w^iX^4.
\]
The $n$-dimensional codeword is obtained by evaluating $p(X)$ in $n$ elements
of $\FF_{2^4}$, $n\leq 15=\nmax$ by (\ref{eq:boundn}).

For $n=4$, if we evaluate $p(X)$ in $w^i$, $i=0,1,2,3$, then the 4
encoded fragments $p(1),p(w),p(w^2),p(w^3)$ are linearly independent and there
is no self-repair possible.

Now for $n=7$, and say, $1,w,w^2,w^4,w^5,w^8,w^{10}$, we get:
\[
(p(1),p(w),p(w^2),p(w^4),p(w^5),p(w^8),p(w^{10})).
\]
Note that
\[
\begin{array}{ll}
w^4 =  w + 1       & w^{10}=w^2+w+1 \\
w^5 =  w^2 + w     &w^{11}=w^3+w^2+w\\
w^6 =  w^3 + w^2   &w^{12}=w^3+w^2+w+1\\
w^7 =  w^3 + w +1  &w^{13}=w^3+w^2+1\\
w^8 =  w^2+1       &w^{14}=w^3+1\\
w^9 =  w^3+w       &w^{15}=1.
\end{array}
\]
Suppose node 5 which stores $p(w^5)$ goes offline. A new comer can get
$p(w^5)$ by asking for $p(w^2)$ and $p(w)$, since
\[
p(w^5)=p(w^2+w)=p(w^2)+p(w).
\]
Table \ref{tab:enumerate} shows other examples of missing fragments and which
pairs can reconstruct them, depending on if 1, 2, or 3 fragments are missing
at the same time.

\begin{table}
\begin{tabular}{|c|c|}
  \hline
  % after \\: \hline or \cline{col1-col2} \cline{col3-col4} ...
  missing &pairs to reconstruct missing fragment(s)\\fragment(s)&\\
  \hline
  $p(1)$   & $(p(w),p(w^4))$;$(p(w^2),p(w^8))$;$(p(w^5),p(w^{10}))$\\
  $p(w)$   & $(p(1),p(w^4))$;$(p(w^2),p(w^5))$;$(p(w^8),p(w^{10}))$\\
  $p(w^2)$ & $(p(1),p(w^8))$;$(p(w),p(w^5))$;$(p(w^4),p(w^{10}))$\\
  \hline
  $p(1)$ and & $(p(w^2),p(w^8))$ or $(p(w^5),p(w^{10}))$ for $p(1)$\\
  $p(w)$     & $(p(w^8),p(w^{10}))$ or $(p(w^2),p(w^{5}))$ for $p(w)$\\
  \hline
  $p(1)$ and & $(p(w^5),p(w^{10}))$ for $p(1)$\\
  $p(w)$ and & $(p(w^8),p(w^{10}))$ for $p(w)$\\
  $p(w^2)$   & $(p(w^4),p(w^{10}))$ for $p(w^2)$  \\
  \hline
\end{tabular}
\vspace{1mm}
\caption{Ways of reconstructing missing fragment(s) in Example
\ref{ex:complete}}\vspace{-6mm}
\label{tab:enumerate}
\end{table}
As for decoding, since $p(X)$ is of degree 5, a node that wants to recover the
data needs $k=3$ linearly independent fragments, say $p(w),p(w^2),p(w^3)$, out
of which it can generate $p(aw+bw^2+cw^3)$, $a,b,c \in \{0,1\}$.
Out of the $7$ non-zero coefficients, 5 of them are enough to recover $p$.
\end{ex}

As shown in the above example, given $k$ fragments, there are different values
of $n$ up to $\nmax$, and different choices of $\{\alpha_1,\ldots,\alpha_n\}$
that can be chosen to define a self-repairing code. We will focus on choosing
the set of $\alpha_i$ to form a subspace of $\FF_{\nmax}$, choice which results
in a particularly nice symmetric structure of the code, namely an XOR-like
structure. However, it is worth repeating that though the encoded fragments
can be obtained as XORs of each other, each fragment is actually containing
information about all the different fragments, which is very different than
a simple XOR of the data itself.
From now on, we will refer to this code as {\em Homomorphic SRC}, and will
write $\hsrc(n,k)$ to emphasize the
code parameters. The analysis that follows refers to this family of
self-repairing codes.

%*************************************************************************%
%
% STATIC RESILIENCE ANALYSIS
%
%*************************************************************************%
\section{Static Resilience Analysis}
\label{sec:static}

The rest of the paper is dedicated to the analysis of the proposed homomorphic
self-repairing codes.
\emph{Static resilience} of a distributed storage system is defined as the probability that an object, once stored in the system, will continue to stay available without any further maintenance, even when a certain fraction of individual member nodes of the distributed system become unavailable.
We start the evaluation of the proposed scheme with a static resilience analysis, where we study how a stored object can be recovered using HSRCs, compared with traditional erasure codes, prior to considering the maintenance process, which will be done in Section \ref{sec:dynamic}.

Let $\probup$ be the probability that any specific node is available. Then, under the assumptions that node availability is $i.i.d$, and no two fragments of the same object are placed on any same node, we can consider that the availability of any fragment is also $i.i.d$ with probability
$\probup$.

%**************************************************************************%
\subsection{A network matrix representation}

Recall that using the above coding strategy, an object $\ov$ of length $M$ is decomposed
into $k$ fragments of length $M/k$:
\[
\ov=(\ov_1,\ldots,\ov_k),~\ov_i\in\FF_{2^{M/k}},
\]
which are further encoded into $n$ fragments of same length:
\[
\pv=(\pv_1,\ldots,\pv_n),~\pv_i\in\FF_{2^{M/k}},
\]
each of the encoded fragment $\pv_i$ is given to a node to be stored.
We thus have $n$ nodes each possessing a binary vector of length $M/k$,
which can be represented as an $n \times M/k$ binary matrix
\begin{equation}\label{eq:MM}
\MM=
\left(
\begin{array}{c}
\pv_1\\
\vdots\\
\pv_n
\end{array}
\right)
=
\left(
\begin{array}{ccc}
p_{1,1} & \ldots & p_{1,M/k}\\
& \vdots & \\
p_{n,1}& \ldots & p_{n,M/k}
\end{array}
\right)
\end{equation}
with $p_{i,j}\in\FF_2$.

\begin{ex}\label{ex:M}
In Example \ref{ex:complete}, we have for $n=4$ that $\MM=I_4$, the 4-dimensional identity matrix,
while for $n=7$, it is
\[
\MM^T=
\left(
\begin{array}{ccccccc}
1&0&0&1&0&1&1\\
0&1&0&1&1&0&1\\
0&0&1&0&1&1&1\\
0&0&0&0&0&0&0\\
\end{array}
\right).
\]
\end{ex}
Thus unavailability of a random node is equivalent to losing a random row of the matrix $\mathbb{M}$. If multiple random nodes (say $n-x$) become unavailable, then the remaining $x$ nodes provide $x$ encoded fragments,
which can be represented by a $x \times M/k$ sub-matrix $\mathbb{M}_x$
of $\MM$. For any given combination of such $x$ available encoded fragments, the original object can still be reconstructed if we can obtain at least $k$ linearly independent rows of $\MM_x$. This is equivalent to say that the object can be reconstructed if the rank of the matrix
$\mathbb{M}_x$ is larger than or equal to $k$.

%*************************************************************************%
\subsection{Probability of object retrieval}

Consider a $(2^d-1) \times d$ binary matrix for some $d>1$, with distinct rows, no all zero row, and thus rank $d$.
If we remove some of the rows uniformly randomly with some probability
$1-\probup$, then we are left with a $x \times d$ sub-matrix - where $x$ is binomially distributed. We define $R(x,d,r)$ as the number of
$x \times d$ sub-matrices with rank $r$, voluntarily including all the possible permutations of the rows in the counting.

\begin{lem}\label{lem:count}
Let $R(x,d,r)$ be the number of $x\times d$ sub-matrices
with rank $r$ of a tall $(2^d-1)\times d$ matrix of rank $d$.
We have that $R(x,d,r)=0$ when (i) $r=0$, (ii) $r>x$, (iii) $r=x$, with $x>d$, or (iv) $r<x$ but $r>d$.
Then, counting row permutations:
\[
R(x,d,r)=
\prod_{i=0}^{r-1}(2^d-2^i)\mbox{ if }r=x,x\leq d,
\]
and for $r< x$ with $r\leq d$:
\[
R(x,d,r)=
R(x-1,d,r-1)(2^d-2^{r-1})+R(x-1,d,r)(2^r-x).
\]
\end{lem}
\begin{proof}
There are no non-trivial matrix with rank $r=0$.
When $r>x$, $r=x$ with $x>d$, or $r<x$ but $r>d$, $R(x,d,r)=0$ since the rank of a matrix cannot be larger than the smallest of its dimensions.

For the case when $r=x$, with $x\leq d$, we deduce $R(x,d,r)$ as follows. To build a matrix $\MM_x$ of rank $x=r$, the first row can be chosen from any of the $2^d-1$ rows in $\MM$, and the second row should not be a multiple of the first row, which gives $2^d-2$ choices. The third row needs to be linearly independent from the first two rows. Since there are $2^2$ linear combinations of the first two rows, which includes the all zero vector which is discarded, we obtain  $2^d-2^2$ choices. In general, the $(i+1)$st row can be chosen from $2^d-2^i$ options that are linearly independent from the $i$ rows that have already been chosen. We thus obtain $R(x,d,r) = \prod_{i=0}^{r-1}(2^d-2^i)$ for $r=x$, $x\leq d$.

For the case where $r< x$ with $r\leq d$, we observe that $x \times d$ matrices of rank $r$ can be inductively obtained by either (I) adding a linearly independent row to a $(x-1) \times d$ matrix of rank $r-1$, or (II) adding a linearly dependent row to a $(x-1) \times d$ matrix of rank $r$. We use this observation to derive the recursive relation
\[
R(x,d,r)= R(x-1,d,r-1)(2^d-2^{r-1})+R(x-1,d,r)(2^r-x),
\]
where $2^d-1-(2^{r-1}-1)$ counts the number of
linearly independent rows that can be added, and $2^r-1-(x-1)$ is on the
contrary the number of linearly dependent rows.
\end{proof}
We now remove the permutations that we counted in the above analysis by introducing
a suitable normalization.
\begin{cor}\label{cor:rho}
Let $\rho(x,d,r)$ be the fraction of sub-matrices of dimension $x \times d$ with rank $r$
out of all possible sub-matrices of the same dimension. Then
\[
\rho(x,d,r)=\frac{R(x,d,r)}{\sum_{j=0}^dR(x,d,j)}=\frac{R(x,d,r)}{C_{x}^{2^d-1}x!}.
\]
\end{cor}
\begin{proof}
It is enough to notice that there are $C_{x}^{2^d-1}$ ways to choose $x$ rows out of the possible $2^d-1$ options.
The chosen $x$ rows can be ordered in $x!$ permutations.
\end{proof}
We now put together the above results to compute the probability $\objup$ of an object being recoverable when
using an $\hsrc(n,k)$ code to store a length $M$ object made of $k$ fragments encoded into $n$ fragments each
of length $M/k$.
\begin{cor}
Using an $\hsrc(n,k)$, the probability $\objup$ of recovering the object is
\[
\objup = \sum_{x=k}^n \sum_{r=k}^d \rho(x,d,r) C_{x}^{n} \probup^x (1-\probup)^{n-x},
\]
where $d=\log_2{n+1}$.
\end{cor}
\begin{IEEEproof}
If $n=\nmax=2^{M/k}-1$, we apply Lemma \ref{lem:count} and Corollary
\ref{cor:rho} with $d=M/k$.
If $n=2^i-1$, for some integer $i\leq M/k$ such that $n>k$ (otherwise there is no encoding), then $\MM$ has $M/k-i$ columns which are either all zeros or all ones vectors, as shown on Example \ref{ex:M}.
Thus the number of its sub-matrices of rank $r$ is given by applying Lemma \ref{lem:count} on the matrix obtained
by removing these redundant columns.
\end{IEEEproof}
We validate the analysis with simulations, and as can be observed from Figure \ref{fig:simvsana}, we obtain a precise match.

%*************************************************************************%
\subsection{Comparison with standard erasure codes}

Let us compare the storage overhead of the proposed scheme against
standard erasure codes. If we use a $(n,k)$ erasure code, then the probability
that the object is recoverable is:
\[
\objup = \sum_{i=k}^{n}C_{i}^n \probup^i (1-\probup)^{n-i}.
\]
\begin{center}
%\begin{figure}
%\includegraphics[scale=0.5]{}
%\caption{}
%\label{simvsana}
%\end{figure}
\begin{figure}[htbp]
 \subfigure[Validation of the static resilience analysis]{
  \includegraphics[scale=0.5]{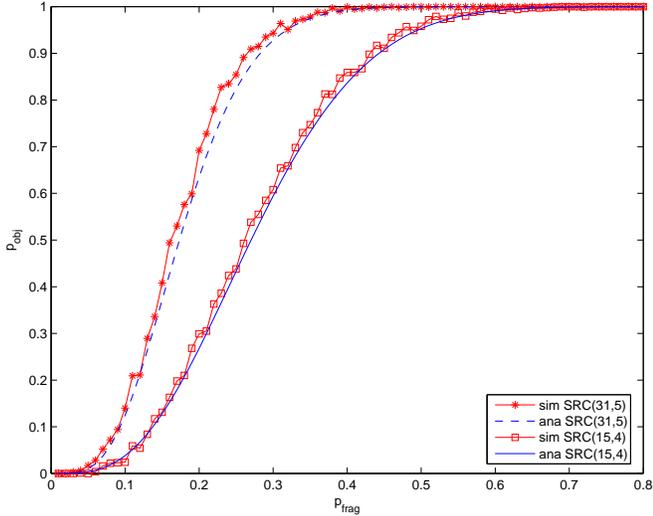}
   \label{fig:simvsana}}
 \subfigure[Comparison of SRC with EC]{
  \includegraphics[scale=0.5]{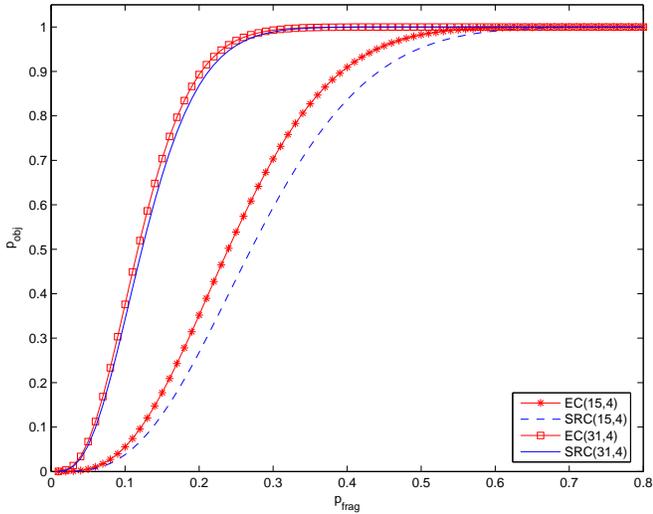}
   \label{fig:ECvsSRCk4}}
 \subfigure[Comparison of SRC with EC]{
  \includegraphics[scale=0.5]{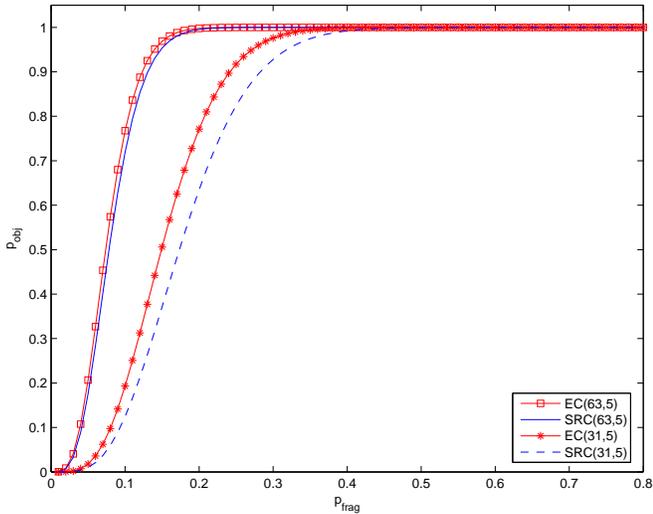}
   \label{fig:ECvsSRCk5}}
 \label{fig:staticresilience}
 \caption{Static resilience of self-repairing codes (SRC): Validation of analysis, and comparison with erasure codes (EC)}
\end{figure}
\end{center}
In Figures \ref{fig:ECvsSRCk4} and \ref{fig:ECvsSRCk5}, we compare the static resilience achieved using the proposed homomorphic SRC with that of traditional erasure codes.

In order to achieve the self-repairing property in SRC, it is obvious that it is necessary to introduce extra `redundancy' in its code structure, but we notice from the comparisons that this overhead is in fact marginal. For the same storage overhead $n/k$, the overall static resilience of SRC is only slightly lower than that of EC, and furthermore, for a fixed $k$, as the value of $n$ increases, SRC's static resilience gets very close to that of EC. Furthermore, even for low storage overheads, with relatively high $\probup$, the probability of object availability is indeed 1. In any storage system, there will be a maintenance operation to replenish lost fragments (and hence, the system will operate for high values of $\probup$). We will further see in the next section that SRCs have significantly lower maintenance overheads. These make SRCs a practical coding scheme for networked storage.

%*************************************************************************%
%
% DYNAMIC ANALYSIS
%
%*************************************************************************%

\section{Communication overheads of self-repair}
\label{sec:dynamic}

In the previous section we studied the probability of recovering an object if it so happens that only $\probup$ fraction of nodes which had originally stored the encoded fragments continue to remain available, while lost redundancy is yet to be replenished. Such a situation may arise either because a lazy maintenance mechanism (such as, in \cite{TotalRecall}) is applied, which triggers repairs only when redundancy is reduced to certain threshold, or else because of multiple correlated failures before repair operations may be carried out. We will next investigate the communication overheads in such scenarios.
Note that this is really the regime in which we need an analysis, since
in absence of correlated failures, and assuming that an eager repair strategy is applied, whenever one encoded block is detected to be unavailable, it is immediately replenished. The proposed HSRC ensures that this one missing fragment can be replenished by obtaining only two other (appropriate) encoded fragments, thanks to the HSRC subspace structure.

\begin{defn}The diversity $\divsrc$ of SRC is defined as the number of mutually exclusive pairs of fragments which can be used to recreate any specific fragment.
\end{defn}

In Example \ref{ex:complete}, it can be seen easily that $\divsrc=3$. Let us assume that $p(w)$ is missing. Any of the three exclusive fragment pairs, namely $((p(1),p(w^4))$; $(p(w^2),p(w^5))$ or $(p(w^8),p(w^{10}))$ may be used to reconstruct $p(w)$. See Table \ref{tab:enumerate} for other examples.

\begin{lem} The diversity $\divsrc$ of a $\hsrc(n,k)$ is $(n-1)/2$.
\end{lem}
\begin{IEEEproof}
We have that $n=2^d-1$ for some suitable $d$.
The polynomial $p(x)$ is evaluated in $\alpha = \sum_{i=0}^{d-1}a_i w^i$, where $a_i \in \{0,1\}$ and $(a_0,...,a_{d-1})$ takes all the possible $2^d$ values, but for the whole zero one. Thus for every $\alpha$, we can create the pairs $(\alpha+\beta,\beta)$ where $\beta$ takes $2^d-2$ possible values, that is all values besides 0 and $\alpha$. This gives $2^d-2$ (which is equal to $n-1$) pairs, but since pairs $(\alpha+\beta,\beta)$ and  $(\beta,\alpha+\beta)$ are equivalent, we have $(n-1)/2$ distinct such pairs.
\end{IEEEproof}

An interesting property of SRC can be inferred from its diversity.
\begin{cor}
\label{cor:2isenough}
For a Homomorphic SRC, if at least $(n+1)/2$ fragments are available, then for any of the unavailable fragments, there exists some pair of available fragments which is adequate to reconstruct the unavailable fragment.
\end{cor}
\begin{IEEEproof}
Consider any arbitrary missing fragment $\alpha$. If up to $(n-1)/2$ fragments were available, in the worst case, these could belong to the $(n-1)/2$ exclusive pairs. However, if an additional fragment is available, it will be paired with one of these other fragments, and hence, there will be at least one available pair with which $\alpha$ can be reconstructed.
\end{IEEEproof}

\subsection{Overheads of recreating one specific missing fragment}

Recall that $x$ is defined as the number of fragments of an object that are available at a given time point. For any specific missing fragment, any one of the corresponding mutually exclusive pairs is adequate to recreate the said fragment. From Corollary \ref{cor:2isenough} we know that if $x \geq (n+1)/2$ then two downloads are enough. Otherwise, we need a probabilistic analysis. Both nodes of a specific pair are available with probability $(x/n)^2$. The probability that only two fragments are enough to recreate the missing fragment is $p_2 = 1-(1- (x/n)^2)^\divsrc$.
%\footnote{Such pairs may be probed either sequentially or in parallel. In all cases, the overheads of pinging are marginal with respect to the actual repair costs, and may also be amortized using the usual regular probing necessary in the system to detect missing fragments. For the sake of completeness: If pinging is done in parallel, the communication overhead is $2 \divsrc$ ping messages. The probability that $i$ pairs need to be pinged sequentially in order to encounter a pair with both nodes online is $(1-(x/n)^2)^{i-1}*(x/n)^2$. The expected number of ping messages will be $\sum_{i=1}^{\divsrc}(1-(x/n)^2)^{i-1}*(x/n)^2$.}

If two fragments are not enough to recreate a specific fragment, it may still be possible to reconstruct it with larger number of fragments. A loose upper bound can be estimated by considering that if 2 fragments are not adequate, $k$ fragments need to be downloaded to reconstruct a fragment,\footnote{Note than in fact, often
fewer than $k$ fragments will be adequate to reconstruct a specific fragment.} which happens with
  a probability $1-p_2 = (1- (x/n)^2)^\divsrc$.%\footnote{A somewhat tighter upper bound can be obtained by using $p_{obj}-p_2$.} %During the probing process, we have also determined which nodes are online, and which are not. So no further probings are required\footnote{Assuming that the time for completing the probes is relatively small w.r.to the level of dynamics in the system.}

%The probability that exactly $y>2$ fragments will be necessary to reconstruct the specific missing fragment for a given value of $x$ is $p_y = ???$. \fixme{If we can not determine this, then we can still find an upper-bound by stating that $k$ blocks are surely enough to recreate a missing block, since $k$-blocks are even enough to recreate the object.}

Thus the expected number $D_x$ of fragments that need to be downloaded to recreate one fragment, when $x$ out of the $n$ encoded fragments are available, can be determined as:
\begin{eqnarray*}
D_x=2 & \mbox{if } x\geq(n+1)/2\\
D_x<2p_2 + k(1-p_2)& \mbox{if } x < (n+1)/2.
\end{eqnarray*}

\subsection{Overhead of recreating all missing fragments}

Above, we studied the overheads to recreate one fragment. All the missing fragments may be repaired, either in parallel (distributed in different parts of the network) or in sequence. If all missing fragments are repaired in parallel, then the total overhead $D_{prl}$ of downloading necessary fragments is: $$D_{prl} = (n-x) D_x.$$

If they are recreated sequentially, then the overhead $D_{seq}$ of downloading necessary fragments is: $$D_{seq} = \sum_{i=x}^{n} D_i.$$

In order to directly compare the overheads of repair for different repair strategies - eager, or lazy parallelized and lazy sequential repairs using SRC, as well as lazy repair with traditional erasure codes, consider that lazy repairs are triggered when a threshold $x=x_{th}$ of available encoded fragments out of $n$ is reached. If eager repair were used for SRC encoded objects, download overhead of $$D_{egr}= 2(n-x_{th})$$ is incurred. Note that, when SRC is applied, the aggregate bandwidth usage for eager repair as well as both lazy repair strategies is the same, assuming that the threshold for lazy repair $x_{th} \geq (n+1)/2$.

In the setting of traditional erasure codes, let us assume that one node downloads enough ($k$) fragments to recreate the original object, and recreates one fragment to be stored locally, and also recreates the remaining $n-x_{th}-1$ fragments, and stores these at other nodes. This leads to a total network traffic: $$D_{EClazy}= k+n-x_{th}-1.$$ Eager strategy using traditional erasure codes will incur $k$ downloads for each repair, which is obviously worse than all the other scenarios, so we ignore it in our comparison.

Note that if less than half of the fragments are unavailable, as observed in Corollary \ref{cor:2isenough}, downloading two blocks is adequate to recreate any specific missing fragment. When too many blocks are already missing, applying a repair strategy analogous to traditional erasure codes, that of downloading $k$ blocks to recreate the whole object, and then recreate all the missing blocks is logical. That is to say, the benefit of reduced maintenance bandwidth usage for SRC (as also of other recent techniques like RGC) only makes sense under a regime when not too many blocks are unavailable. Let us define $x_c$ as the critical value, such that if the threshold for lazy repair in traditional erasure codes $x_{th}$ is less than this critical value, then, the aggregate fragment transfer traffic to recreate missing blocks will be less using the traditional technique (of downloading $k$ fragments to recreate whole object, and then replenish missing fragments) than by using SRC. Recall that for $x \geq (n+1)/2$, $D_{egr}=D_{prl}=D_{seq}$. One can determine $x_c$ as follows. We need $D_{egr} \leq D_{EClazy}$, implying that
\[
 2n-2x_{c} \leq n-1+k-x_{c}\Rightarrow x_c=n+1-k.
\]
Figure \ref{fig:trafficperlostblock} shows the average amount of network traffic to transfer encoded fragments per lost fragment when the various lazy variants of repair are used, namely parallel and sequential repairs with SRC, and (by default, sequential) repair when using EC. The \emph{x-axis} represents the threshold $x_{th}$ for lazy repair, such that repairs are triggered only if the number of available blocks for an object is not more than $x_{th}$. Use of an eager approach with SRC incurs a constant overhead of two fragments per lost block.\footnote{Note that there are other messaging overheads to disseminate necessary meta-information (e.g., which node stores which fragment), but we ignore these in the figure, considering that the objects being stored are large, and data transfer of object fragments dominates the network traffic. This assumption is reasonable, since for small-objects, it is well known that the meta-information storage overheads outweigh the benefits of using erasures, and hence erasures are impractical for small objects.}

\begin{figure}[htbp]
\begin{center}
\includegraphics[scale=0.5]{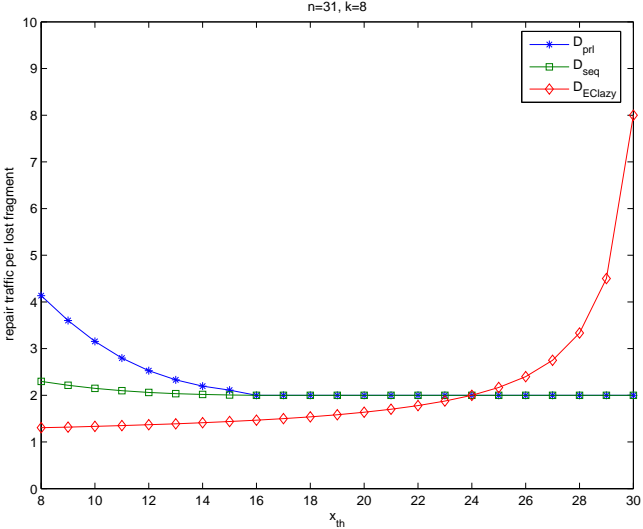}
\caption{Average traffic per lost block for various choices of $x_{th}$}
\label{fig:trafficperlostblock}
\end{center}\vspace{-6mm}
\end{figure}

There are several quantitative and qualitative implications of the above observed behaviors. To start with, we note that an engineering solution like lazy repair which advocates waiting before repairs are triggered, amortizes the repair cost per lost fragment, and is effective in reducing total bandwidth consumption and outperforms SRC (in terms of total bandwidth consumption), provided the threshold of repair $x_{th}$ is chosen to be lower than $x_c$. This is in itself not surprising. However, for many typical choices of $(n,k)$ in deployed systems such as $(16,10)$ in Cleversafe \cite{cleversafe}, or $(517,100)$ in Wuala \cite{wuala}, a scheme like SRC is practical. In the former scenario, $x_c$ is too low, and waiting so long makes the system too vulnerable to any further failures (i.e., poor system health). In the later scenario, that is, waiting for hundred failures before triggering repairs seems both unnecessary, and also, trying to repair 100 lost fragments simultaneously will lead to huge bandwidth spikes. The system's vulnerability to further failures, as well as spiky bandwidth usage are known problems of lazy repair strategies \cite{dattaSSS}.

Using SRC allows for a flexible choice of either an eager or lazy (but with much higher threshold $x_{th}$) approaches to carry out repairs, where the repair cost per lost block stays constant for a wide range of values (up till $x_{th}\geq(n+1)/2$). Such a flexible choice makes it easier to also benefit from the primary advantage of lazy repair in peer-to-peer systems, namely, to avoid unnecessary repairs due to temporary churn, without the drawbacks of (i) having to choose a threshold which leads to system vulnerability or (ii) choose a much higher value of $n$ in order to deal with such vulnerability, and (iii) have spiky bandwidth usage.

\subsection{Fast parallel repairs using SRC: A qualitative discussion}

We observed in the previous section that while SRC is effective in significantly reducing bandwidth usage to carry out maintenance of lost redundancy in coding based distributed storage systems, depending on system parameter choices, an engineering solution like lazy repair while using traditional EC may (or not) outperform SRC in terms of total bandwidth usage, even though using lazy repair with EC entails several other practical disadvantages.

A final advantage of SRC which we further showcase next is the possibility to carry out repairs of different fragments independently and in parallel (and hence, quickly). If repair is not fast, it is possible that further faults occur during the repair operations, leading to both performance deterioration as well as, potentially, loss of stored objects.

Consider the following scenario for ease of exposition: Assume that each node in the storage network has an uplink/downlink capacity of 1 (coded) fragment per unit time. Further assume that the network has relatively (much) larger aggregate bandwidth. Such assumptions correspond reasonably with various networked storage system environments.
%In peer-to-peer systems, individual users have orders of magnitude less bandwidth w.r.to the total bandwidth in the network. In data centers, edge nodes typically have 1GE connections, while the aggregation and core layers have 10GE connections.

Consider that for the Example \ref{ex:complete}, originally $n$ was chosen to be $n_{max}$, that is to say, a $\hsrc(15,3)$ was used. Because of some reasons (e.g., lazy repair or correlated failures), let us say that seven encoded fragments, namely $p(1),\ldots,p(w^6)$ are unavailable while fragments $p(w^7)...p(w^{15})$ are available. Table \ref{tab:enumeratereconstruction} enumerates possible pairs to reconstruct each of the missing fragments.

\begin{table}
\begin{tabular}{|c|c|}
  \hline
  % after \\: \hline or \cline{col1-col2} \cline{col3-col4} ...
  fragment & suitable pairs to reconstruct\\
  \hline
  $p(1)$   & $(p(w^7),p(w^9))$;$(p(w^{11}),p(w^{12}))$\\
  $p(w)$   & $(p(w^7),p(w^{14}))$;$(p(w^8),p(w^{10}))$\\
  $p(w^2)$ & $(p(w^7),p(w^{12}))$;$(p(w^9),p(w^{11}))$;$(p(w^{12}),p(w^{10}))$\\
  $p(w^3)$ & $(p(w^8),p(w^{13}))$;$(p(w^{10}),p(w^{12}))$\\
  $p(w^4)$ & $(p(w^9),p(w^{14}))$;$(p(w^{11}),p(w^{13}))$\\
  $p(w^5)$ & $(p(w^7),p(w^{13}))$;$(p(w^{12}),p(w^{14}))$\\
  $p(w^6)$ & $(p(w^7),p(w^{10}))$;$(p(w^8),p(w^{14}))$\\
  \hline
\end{tabular}
\vspace{1mm}
\caption{Scenario: Seven fragments $p(1),\ldots,p(w^6)$ are missing}
\label{tab:enumeratereconstruction}\vspace{-8mm}
\end{table}

A potential schedule to download the available blocks at different nodes to recreate the missing fragments is as follows: In first time slot, $p(w^{11})$, $p(w^{10})$, $p(w^{12})$, nothing, $p(w^{13})$, $p(w^{7})$ and $p(w^{8})$ are downloaded separately by seven nodes trying to recreate each of $p(1),\ldots,p(w^6)$ respectively. In second time slot $p(w^{12})$, $p(w^{8})$, $p(w^{7})$, $p(w^{10})$, $p(w^{11})$, $p(w^{13})$ and $p(w^{14})$ are downloaded. Note that, besides $p(w^3)$, all the other missing blocks can now already be recreated. In third time slot, $p(w^{12})$ can be downloaded to recreate it. Thus, in this example, six out of the seven missing blocks could be recreated within the time taken to download two fragments, while the last block could be recreated in the next time round, subject to the constraints that any node could download or upload only one block in unit time.

Even if a full copy of the object (hybrid strategy \cite{Liskov}) were to be maintained in the system, with which to replenish the seven missing blocks, it would have taken seven time units. While, if no full copy was maintained, using traditional erasure codes would have taken at least nine time units.

This example demonstrates that SRC allows for fast reconstruction of missing blocks. Orchestration of such distributed reconstruction to fully utilize this potential in itself poses interesting algorithmic and systems research
challenges which we intend to pursue as part of future work.
%Such work can additionally take into account network topology to determine schemes to optimally place the encoded fragments.

\section{Conclusion}
We propose a new family of codes, called self-repairing codes, which are designed by taking into account specifically the characteristics of distributed networked storage systems. Self-repairing codes achieve excellent properties in terms of maintenance of lost redundancy in the storage system, most importantly: (i) low-bandwidth consumption for repairs (with flexible/somewhat independent choice of whether an eager or lazy repair strategy is employed), (ii) parallel and independent (thus very fast) replenishment of lost redundancy. When compared to erasure codes, the self-repairing property is achieved by marginally compromising on static resilience for same storage overhead, or conversely, utilizing marginally more storage space to achieve equivalent static resilience. This paper provides the theoretical foundations for SRCs, and shows its potential benefits for distributed storage. There are several algorithmic and systems research challenges in harnessing SRCs in distributed storage systems, e.g., design of efficient decoding algorithms, or placement of encoded fragments to leverage on network topology to carry out parallel repairs, which are part of our ongoing and future work.

%SIMULATIONS:
%* range: one object = up to 64 megabytes
%* example from wala: (517,100)=(n,k) code
%* R-S: example from wiki, typically k=223 over F8 (=8 bit symbols)

%************************************************************************%
%
% ACK
%
%***********************************************************************%

%\section*{Acknowledgment}

%The research of F. Oggier is supported in part by the Singapore National
%Research Foundation under Research Grant NRF-RF2009-07 and NRF-%CRP2-2007-03,
%and in part by the Nanyang Technological University under Research
%Grant M58110049 and M58110070.

%********************************************************************%
%
% BIBLIO
%
%********************************************************************%

\end{document}